\documentclass[pra,aps,twoside,superscriptaddress]{revtex4}

\usepackage{amsmath,amsfonts,amssymb,caption,hyperref,color,epsfig,graphics,graphicx,latexsym,mathrsfs,revsymb,theorem,url,verbatim,epstopdf}

\hypersetup{colorlinks,linkcolor={blue},citecolor={blue},urlcolor={red}} 


\usepackage{hyperref}
\usepackage{qcircuit}
\newtheorem{definition}{Definition}
\newtheorem{proposition}[definition]{Proposition}
\newtheorem{lemma}[definition]{Lemma}

\newtheorem{theorem}[definition]{Theorem}
\newtheorem{corollary}[definition]{Corollary}
\newtheorem{conjecture}[definition]{Conjecture}

\newtheorem{remark}[definition]{Remark}
\newtheorem{example}[definition]{Example}

\def\squareforqed{\hbox{\rlap{$\sqcap$}$\sqcup$}}
\def\qed{\ifmmode\squareforqed\else{\unskip\nobreak\hfil
\penalty50\hskip1em\null\nobreak\hfil\squareforqed
\parfillskip=0pt\finalhyphendemerits=0\endgraf}\fi}
\def\endenv{\ifmmode\;\else{\unskip\nobreak\hfil
\penalty50\hskip1em\null\nobreak\hfil\;
\parfillskip=0pt\finalhyphendemerits=0\endgraf}\fi}
\newenvironment{proof}{\noindent \textbf{{Proof.~} }}{\qed}
\def\Dbar{\leavevmode\lower.6ex\hbox to 0pt
{\hskip-.23ex\accent"16\hss}D}
\makeatletter
\def\url@leostyle{%
  \@ifundefined{selectfont}{\def\UrlFont{\sf}}{\def\UrlFont{\small\ttfamily}}}
\makeatother
\def\bcj{\begin{conjecture}}
\def\ecj{\end{conjecture}}
\def\bcr{\begin{corollary}}
\def\ecr{\end{corollary}}
\def\bd{\begin{definition}}
\def\ed{\end{definition}}
\def\bea{\begin{eqnarray}}
\def\eea{\end{eqnarray}}
\def\bem{\begin{enumerate}}
\def\eem{\end{enumerate}}
\def\bex{\begin{example}}
\def\eex{\end{example}}
\def\bim{\begin{itemize}}
\def\eim{\end{itemize}}
\def\bl{\begin{lemma}}
\def\el{\end{lemma}}
\def\bma{\begin{bmatrix}}
\def\ema{\end{bmatrix}}
\def\bpf{\begin{proof}}
\def\epf{\end{proof}}
\def\bpp{\begin{proposition}}
\def\epp{\end{proposition}}
\def\bqu{\begin{question}}
\def\equ{\end{question}}
\def\br{\begin{remark}}
\def\er{\end{remark}}
\def\bt{\begin{theorem}}
\def\et{\end{theorem}}

\def\btb{\begin{tabular}}
\def\etb{\end{tabular}}

\newcommand{\nc}{\newcommand}

 \nc{\bbA}{\mathbb{A}} \nc{\bbB}{\mathbb{B}} \nc{\bbC}{\mathbb{C}}
 \nc{\bbD}{\mathbb{D}} \nc{\bbE}{\mathbb{E}} \nc{\bbF}{\mathbb{F}}
 \nc{\bbG}{\mathbb{G}} \nc{\bbH}{\mathbb{H}} \nc{\bbI}{\mathbb{I}}
 \nc{\bbJ}{\mathbb{J}} \nc{\bbK}{\mathbb{K}} \nc{\bbL}{\mathbb{L}}
 \nc{\bbM}{\mathbb{M}} \nc{\bbN}{\mathbb{N}} \nc{\bbO}{\mathbb{O}}
 \nc{\bbP}{\mathbb{P}} \nc{\bbQ}{\mathbb{Q}} \nc{\bbR}{\mathbb{R}}
 \nc{\bbS}{\mathbb{S}} \nc{\bbT}{\mathbb{T}} \nc{\bbU}{\mathbb{U}}
 \nc{\bbV}{\mathbb{V}} \nc{\bbW}{\mathbb{W}} \nc{\bbX}{\mathbb{X}}
 \nc{\bbZ}{\mathbb{Z}}

 \nc{\bA}{{\bf A}} \nc{\bB}{{\bf B}} \nc{\bC}{{\bf C}}
 \nc{\bD}{{\bf D}} \nc{\bE}{{\bf E}} \nc{\bF}{{\bf F}}
 \nc{\bG}{{\bf G}} \nc{\bH}{{\bf H}} \nc{\bI}{{\bf I}}
 \nc{\bJ}{{\bf J}} \nc{\bK}{{\bf K}} \nc{\bL}{{\bf L}}
 \nc{\bM}{{\bf M}} \nc{\bN}{{\bf N}} \nc{\bO}{{\bf O}}
 \nc{\bP}{{\bf P}} \nc{\bQ}{{\bf Q}} \nc{\bR}{{\bf R}}
 \nc{\bS}{{\bf S}} \nc{\bT}{{\bf T}} \nc{\bU}{{\bf U}}
 \nc{\bV}{{\bf V}} \nc{\bW}{{\bf W}} \nc{\bX}{{\bf X}}
 \nc{\bZ}{{\bf Z}}

\nc{\cA}{{\cal A}} \nc{\cB}{{\cal B}} \nc{\cC}{{\cal C}}
\nc{\cD}{{\cal D}} \nc{\cE}{{\cal E}} \nc{\cF}{{\cal F}}
\nc{\cG}{{\cal G}} \nc{\cH}{{\cal H}} \nc{\cI}{{\cal I}}
\nc{\cJ}{{\cal J}} \nc{\cK}{{\cal K}} \nc{\cL}{{\cal L}}
\nc{\cM}{{\cal M}} \nc{\cN}{{\cal N}} \nc{\cO}{{\cal O}}
\nc{\cP}{{\cal P}} \nc{\cQ}{{\cal Q}} \nc{\cR}{{\cal R}}
\nc{\cS}{{\cal S}} \nc{\cT}{{\cal T}} \nc{\cU}{{\cal U}}
\nc{\cV}{{\cal V}} \nc{\cW}{{\cal W}} \nc{\cX}{{\cal X}}
\nc{\cZ}{{\cal Z}}

\nc{\hA}{{\hat{A}}} \nc{\hB}{{\hat{B}}} \nc{\hC}{{\hat{C}}}
\nc{\hD}{{\hat{D}}} \nc{\hE}{{\hat{E}}} \nc{\hF}{{\hat{F}}}
\nc{\hG}{{\hat{G}}} \nc{\hH}{{\hat{H}}} \nc{\hI}{{\hat{I}}}
\nc{\hJ}{{\hat{J}}} \nc{\hK}{{\hat{K}}} \nc{\hL}{{\hat{L}}}
\nc{\hM}{{\hat{M}}} \nc{\hN}{{\hat{N}}} \nc{\hO}{{\hat{O}}}
\nc{\hP}{{\hat{P}}} \nc{\hR}{{\hat{R}}} \nc{\hS}{{\hat{S}}}
\nc{\hT}{{\hat{T}}} \nc{\hU}{{\hat{U}}} \nc{\hV}{{\hat{V}}}
\nc{\hW}{{\hat{W}}} \nc{\hX}{{\hat{X}}} \nc{\hZ}{{\hat{Z}}}

\nc{\hn}{{\hat{n}}}

\def\dim{\mathop{\rm Dim}}

\def\max{\mathop{\rm max}}
\def\min{\mathop{\rm min}}

\def\supp{\mathop{\rm supp}}
\def\tr{\mathop{\rm Tr}}

\newcommand{\bra}[1]{\langle#1|}
\newcommand{\ket}[1]{|#1\rangle}

\def\Dbar{\leavevmode\lower.6ex\hbox to 0pt
{\hskip-.23ex\accent"16\hss}D}

\begin{document}
	\title{The entropy of quantum causal networks}
	\author{Xian Shi}\email[]
{shixian01@buaa.edu.cn}
\affiliation{School of Mathematical Sciences, Beihang University, Beijing 100191, China}
\author{Lin Chen}\email[]{linchen@buaa.edu.cn (corresponding author)}
\affiliation{School of Mathematical Sciences, Beihang University, Beijing 100191, China}
\affiliation{International Research Institute for Multidisciplinary Science, Beihang University, Beijing 100191, China}
	\begin{abstract}
\indent Quantum networks play a key role in many scenarios of quantum information theory. Here we consider the quantum causal networks in the manner of entropy. First we present a revised smooth max-relative entropy of quantum combs, then we present a lower and upper bound of a type \uppercase\expandafter{\romannumeral2} error of the hypothesis testing.  Next we present a lower bound of the smooth max-relative entropy for the quantum combs with asymptotic equipartition. At last, we consider the score to quantify the performance of an operator. We present a quantity equaling to the smooth asymptotic version of the performance of a quantum positive operator.
	\end{abstract}
\maketitle
\section{Introduction}
\indent Nowadays, quantum networks attract much attention of researchers from the technological levels and theoretical levels. Technologically, The development of quantum communication \cite{wang2015quantum,pirandola2015advances} and computation \cite{llewellyn2020chip} prompt the realization of the quantum network \cite{chiribella2013quantum,gutoski2018fidelity}. Theoretically, quantum networks provide a framework for quantum games \cite{gutoski2007toward}, the discrimination and transformation of quantum channels \cite{chiribella2008memory,lloyd2011quantum,chiribella2012perfect}. They also prompt the advances of models of quantum computation \cite{bisio2010optimal}.\\
\indent One of the fundamental concepts in quantum information theoy is the relative entropy for two quantum states . In 1962, Umegaki introduced the quantum relative entropy \cite{umegaki1954conditional}, then the relative entropy was extended in terms of Renyi entropy \cite{renyi1961measures},  min- and max-entropy \cite{datta2009min}. Here the max-relative entropy was shown to interpret a number of operational tasks \cite{brandao2011one,napoli2016robustness,anshu2018quantifying,takagi2019general,seddon2020quantifying}. The other approach to study the relative entropy is in the manner of smooth entropy \cite{renner2008security}. The smooth max-relative entropy can also be used to interpret the one-shot cost tasks in terms of entanglement and coherence frameworks \cite{brandao2011one,zhu2017operational}. Another smooth entropy is correlated with quantum hypothesis testing. This is a fundamental task in statistics theory, that is, an experimentalist should make a decision of a binary testing. The aim is to present an optimal strategy to minimize the error possibility. This task correponds to the discrimination of quantum states \cite{ogawa2005strong,li2014second} and channels \cite{cooney2016strong, gour2019quantify}. It also provides a way to compute the capacity of quantum channels \cite{buscemi2010quantum,anshu2018hypothesis}.  Recently, some work were done on the quantum channel in the manner of  the entroy \cite{gour2019quantify,fang2020no}. As far as we know, the quantum networks are not much studied.\\
\indent Quantum networks can be seen as transformations from states to channels, from channels to channels, and from sequences of channels to channels. In \cite{chiribella2009theoretical}, the authors presented a framework of quantum networks, there they also introduced two important concepts, quantum combs and link product. Due to the importance to study the optimalization of quantum networks,  \cite{chiribella2016optimal} proposed a semidefinite programming method for this problem. There the authors presented the measure of a quantum performance  correponds to the max relative entropy, as well as they also present an opertational interpretation of the max relative entropy of quantum combs.\\
\indent In this paper,  we first present a quantity similar to the smooth max-relative entropy of quantum combs $C_0$ with respect to $C_1.$ We also give a bound between the two max-relative entropies of quantum combs. Then we present a relation between the type \uppercase\expandafter{\romannumeral2} error of quantum hypothesis testing and the smooth max-relative entropy of quantum combs we defined. At last, we present a smooth asymptotic version of the performance of a quantum network.\\
\indent This paper is organized as follows.  In section \uppercase\expandafter{\romannumeral2}, we recall the preliminary knowledge needed. In section \uppercase\expandafter{\romannumeral3}, we present the main results. First we present the bound between the type \uppercase\expandafter{\romannumeral2} error of quantum hypothesis testing on quantum combs. Then we present a quantity equaling to the smooth asymptotic version of the performance of a quantum positive operator. In section 
 \uppercase\expandafter{\romannumeral4}, we end with a conclusion.
 \section{Preliminary knowledge}
In this section, we first recall the definition and some properties of the quantum operations and link products. We next recall the definition and properties of quantum causal networks. At last, we recall the maximal relative entropy and present a revised type \uppercase\expandafter{\romannumeral2} error of hypothesis testing on quantum combs we will consider in the following.
\subsection{Linear Maps and Link Products}
\indent In this article, we denote $L(\mathcal{H}_0,\mathcal{H}_1)$ as the set of linear operators from a finite dimensional Hilbert space $\mathcal{H}_0$ to $\mathcal{H}_1,$ $L(\mathcal{H})$ as the set of linear operators from a Hilbert space $\mathcal{H}$ to $\mathcal{H}$, and $Pos(\mathcal{H})$ as the set of positive semidefinite operators on $\mathcal{H}$. When $M$ is a positive semidefinite operator on a Hilbert space $\mathcal{H},$ we also denote $\lambda_{\max}(M)$ as the maximal eigenvalue of $M.$ When $S, T\in Pos(\mathcal{H}),$ $\{S\ge T\}$ is the projection operator on the space $span\{\ket{v}|\bra{v}(S-T)\ket{v}\ge 0\}$.\\
\indent Assume $\mathcal{C}$ is a map from operators on a Hilbert space $\mathcal{H}_A$ to that on $\mathcal{H}_B.$ When $\mathcal{C}$ is completely positive trace non-increasing,  \vspace{2mm} then we say $\mathcal{C}$ is a quantum operation. When $\mathcal{C}$ is completely positive trace-preserving, then we say $\mathcal{C}$ is a quantum channel. In the following, we use the diagrammatic notation
\Qcircuit@C=1.5em @R=1.2em {
	&\ustick{A} \qw&\gate{\mathcal{C}}&\ustick{B}\qw&\qw}, here $\mathcal{C}$ in the above diagram is of type $A\rightarrow B.$\\

\indent  Next we recall the isomorphism between a quantum operation $\mathcal{M}$ in $L(L(\mathcal{H}_0),L(\mathcal{H}_1))$ and an operator $M$ in $L(L(\mathcal{H}_1\otimes\mathcal{H}_0)).$ 
\begin{lemma}(Choi-Jamiolkowski (C-J) isomorphism)\cite{choi1975completely}
	The bijective correspondence $\mathcal{L}: \mathcal{M}\rightarrow M$ is defined as follows:
	\begin{align}
	M=\mathcal{L}(\mathcal{M})=\mathcal{M}\otimes I_{\mathcal{H}_1}(\ket{\Psi}\bra{\Psi}).
	\end{align}
	here we denote $\ket{\Psi}=\sum_i\ket{i}_{\mathcal{H}_0}\ket{i}_{\mathcal{H}_1},$ $\{\ket{i}_{\mathcal{H}_0}\},\{\ket{i}_{\mathcal{H}_1}\}$ are the orthnormal bases of the Hilbert space $\mathcal{H}_0$ and $\mathcal{H}_1.$
	Its inverse is defined as 
	\begin{align}\label{im}
	[\mathcal{L}^{-1}(M)](X)=tr_{\mathcal{H}_0}[(I_{\mathcal{H}_1}\otimes X^T)M]
	\end{align}\\
\end{lemma}\par
In the following, we denote $M$ as $\mathcal{L}(M).$ Next we present some facts needed, readers who are interested in the proof of these results please refer to \cite{chiribella2009theoretical}.
\begin{lemma}
(i).	A linear map $\mathcal{M}\in L(\mathcal{H}_0,\mathcal{H}_1)$ is trace preserving if and only if its C-J operator satisfies the following equality
	\begin{align}
	\tr_{\mathcal{H}_1}M=I_{\mathcal{H}_0}.
	\end{align}
(ii).	A linear map $\mathcal{M}$ is Hermitian preserving if and only if its C-J operator is Hermitian.\\
(iii).	A linear map $\mathcal{M}$ is completely positive if and only if its C-J operator is positive semidefinite.
\end{lemma}\par
\indent Two maps $\mathcal{M}$ and $\mathcal{N}$ are composed if the output space of $\mathcal{M}$ is the input space of $\mathcal{N}.$ If $\mathcal{M}:L(\mathcal{H}_0)\rightarrow L(\mathcal{H}_1),$ $\mathcal{N}:L(\mathcal{H}_1)\rightarrow L(\mathcal{H}_2),$ next we denote $\mathcal{C}=\mathcal{N}\circ\mathcal{M},$ then
\begin{align}
C(X)=&\tr_{\mathcal{H}_1}[(I_{\mathcal{H}_2}\otimes \tr_{\mathcal{H}_0}[(I_{\mathcal{H}_1}\otimes X^T)M]^T)N]\nonumber\\
=&\tr_{\mathcal{H}_1,\mathcal{H}_0}[(I_{\mathcal{H}_2}\otimes I_{\mathcal{H}_1}\otimes X)(I_{\mathcal{H}_2}\otimes M^{T_1})(N\otimes I_{\mathcal{H}_0})],
\end{align} 
comparing with the equality (\ref{im}), we have 
\begin{align}
C=\tr_{\mathcal{H}_1}[(I_{\mathcal{H}_2}\otimes M^{T_1})(N\otimes I_{\mathcal{H}_0})],
\end{align}
then we denote 
\begin{align}\label{lp}
N*M=\tr_{\mathcal{H}_1}[(I_{\mathcal{H}_2}\otimes M^{T_1})(N\otimes I_{\mathcal{H}_0})],
\end{align}
here we denote (\ref{lp}) as the link product of $N$ and $M.$\\
\indent Next we present some properties of the link product. Readers who are interested in the proof of these properties please refer to \cite{chiribella2009theoretical}.
\begin{lemma}\label{lpl}
	1. If $M$ and $N$ are Hermitian, then $M*N$ is Hermitian.\\
	2. If $M$ and $N$ are positive, then $M*N$ is positive.\\
	3. The link product is associative, $i. e.$ $A*(B*C)=(A*B)*C.$\\
	4. $N*M=SWAP_{\mathcal{H}_0,\mathcal{H}_2}(M*N)SWAP_{\mathcal{H}_0,\mathcal{H}_2},$ here $SWAP_{\mathcal{H}_0,\mathcal{H}_2}$ is the unitary operator that swaps the $\mathcal{H}_0$ and $\mathcal{H}_2.$
\end{lemma}

\subsection{Quantum Networks}
\indent A quantum network is a collection of quantum devices connected with each other. If there are no loops connecting the output of a device to the output of the same device, the quantum network is causal. A network is deterministic if all devices are channels. By the interpretation in \cite{chiribella2009theoretical}, a quantum causal network can always be represented as an ordered sequence of quantum devices, such as Fig. \ref{qc}.
\begin{figure}
	\centering
\includegraphics[width=100mm]{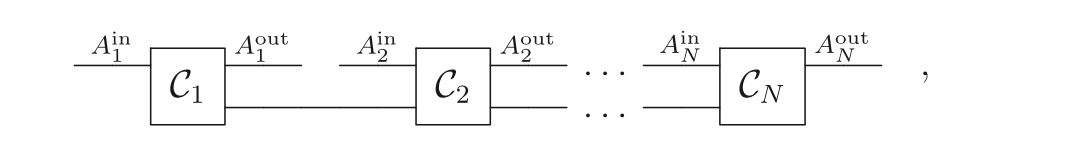}
\caption{Here we plot an example of quantum causal network}\label{qc}
\end{figure}
Here $A_i^{in}(A_i^{out})$ is the input (output) system of the network at the $i$-th step. Assume the C-J operator of  $\mathcal{C}_i$ is $C_i$, then the C-J operator of the above network $(\ref{qc})$ is 
\begin{align}\label{qcc}
C=C_1*C_2*\cdots*C_N.
\end{align}
The C-J operator of a deterministic network is called a quantum comb \cite{chiribella2009theoretical}, it is a positive operator on $\otimes_{j=1}^N(\mathcal{H}_j^{out}\otimes\mathcal{H}_j^{in}),$ here $\mathcal{H}_j^{out}(\mathcal{H}_j^{in})$ is the Hilbert system of $A_j^{out}(A_j^{in})$. Next we recall a result on charactering the quantum combs.
\begin{lemma}\cite{chiribella2009theoretical}\label{qcb}
	A positive operator $C$ is a quantum comb if and only if 
	\begin{align}
	\tr_{A_n^{out}} C^{(n)}=I_{A_n^{in}}\otimes C^{(n-1)}, \hspace{3mm}\forall n\in{1,2,\cdots, N}\label{qc'},
	\end{align}
	here $C^{(n)}$ is a suitable operator on $\mathcal{H}_n=\otimes_{j=1}^n(\mathcal{H}_j^{out}\otimes\mathcal{H}_j^{in}),$ $C^{(N)}=C,$ $C^{(0)}=1.$
\end{lemma}\par
\indent A quantum network $\mathcal{T}$ is non-deterministic if there exists the C-J operator $R$ of a quantum deterministic network $\mathcal{R}$, $R\ge T.$ An example of non-deterministic network is the Fig. \ref{f1}, here $\rho$ is a quantum state, $\mathcal{D}_i$ are quantum channels, $\{P_x\}_x$ is a positive operator-valued measure (POVM). The network of type in Fig. \ref{f1} can be used to probe the network of type  (\ref{qc}), then the network can be represented as Fig. \ref{f2}, then the probability of the outcome $x$ is 
\begin{align}
p_x=\rho*C_1*D_1*C_2*D_2*\cdots*D_{N-1}*C_N*P_x^T\label{px}
\end{align}
by the Lemma \ref{lpl}, when we omit the $SWAP$ operation, that is, we assume the Hilbert spaces have been reordered we need, then $(\ref{px})$ can be written as
\begin{align}
p_x=&(\rho*D_1*D_2*\cdots*D_{N-1}*P_x^T)*(C_1*C_2*\cdots*C_{N})\nonumber\\
=&T_x*C\nonumber\\
\end{align}
here 
\begin{align}
T_x=\rho*D_1*D_2*\cdots*D_{N-1}*P_x^T,
\end{align}
we call the set of operations $\boldsymbol{T}=\{T_x\}_x$ a quantum tester.  Next we recall the mathematical structure of a quantum tester, which was given in \cite{chiribella2009theoretical}.
\begin{lemma}\cite{chiribella2009theoretical}\label{qut}
	Let $\boldsymbol{T}$ be a collection of positive operations on $\otimes_{j=1}^N(\mathcal{H}_j^{out}\otimes\mathcal{H}_j^{in}).$ $\boldsymbol{T}$ is a quantum tester if and only if
	\begin{align}
	\sum_{x\in X} T_x=&I_{A^{out}_N}\otimes \Gamma^{(N)}\label{qt}\\
	\tr_{A_n^{in}}[\Gamma^{(n)}]=&I_{A_{n-1}^{out}}\otimes\Gamma^{(n-1)},\hspace{3mm}n=2,3,\cdots,N.\label{qt1}\\
	\tr_{A_1^{in}}[\Gamma^{(1)}]=&1,\label{qt2}
	\end{align}
	here $\Gamma^{(n)},$ $n=1,2,\cdots,N$ is a positive operator on $\mathcal{H}^{in}_n\otimes[\otimes_{j=1}^n(\mathcal{H}_j^{out}\otimes \mathcal{H}_j^{in})]$.
\end{lemma}
 
\begin{figure*}
	\centering
	\includegraphics[width=170mm]{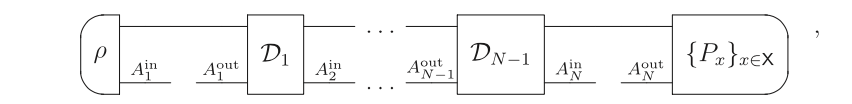}
	\caption{Here we present an example of a non-deterministic network.}	\label{f1}
\end{figure*}
\begin{figure*}
	\centering
	\includegraphics[width=170mm]{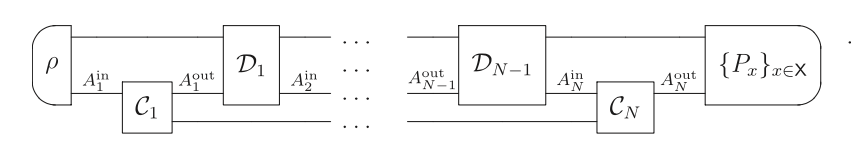}
	\caption{Here we present the type of the network when we probe the quantum comb $(\ref{qc})$ by network of Fig. $\ref{f1}$.}	\label{f2}
\end{figure*}

\subsection{Max-Relative Entropy and Hypothesis Testing}
\indent Here we first recall the definition of max-relative entropy.
\begin{definition}
	Assume $M$ and $N$ are two quantum combs on a Hilbert space $\otimes_{j=1}^n(\mathcal{H}_j^{out}\otimes \mathcal{H}_j^{in}),$ the max entropy of $M$ relative to $N$ is defined as
	\begin{align}
	D_{\max}(M||N)=\log\min\{\lambda|\lambda N-M\ge 0\},
	\end{align}
	when $\supp(M)\not\subset \supp(N),$ $D_{\max}(M||N)=\infty.$\\
	Assume $\epsilon\in (0,1),$ the $\epsilon-$smooth max-relative of $M$ and $N$ is defined as 
	\begin{align}
	D_{\max}^{\epsilon}(M||N)=\min_{\tilde{M}\in S} D_{\max}(\tilde{M}||N),
	\end{align} 
	where the minimum takes over all $\tilde{M}$ in the set $S=\{\tilde{M}|1/2||\tilde{M}-M||\le \epsilon, \tilde{M} \textit{is  a quantum comb.} \}.$
\end{definition}
Next we recall the quantum hypothesis testing on states.  Assume 
\begin{align}
\textit{Null hypothesis     } H_0:\hspace{3mm}\rho\nonumber\\
\textit{Alternative hypothesis } H_1:\hspace{3mm \sigma},
\end{align} 
for a POVM $\{\Pi,I-\Pi\},$ the error probability of type \MakeUppercase{\romannumeral1} and the error probability of type \MakeUppercase{\romannumeral2} are defined as
\begin{align}
\alpha(\MakeUppercase{\romannumeral1})=&\tr[(I-\Pi)\rho],\nonumber\\
\beta(\MakeUppercase{\romannumeral2})=&\tr[\Pi\sigma].
\end{align} 
here $\alpha(\MakeUppercase{\romannumeral1})$ is the probability of accepting $\sigma$ when $\rho$ is true, and $\beta$(\MakeUppercase{\romannumeral2}) is the probability of accepting $\rho$ when $\sigma$ is true. Next we recall the quantity $\beta_{\epsilon}(\rho||\sigma),$ 
\begin{align}
\beta_{\epsilon}(\rho||\sigma)=\min\{\beta_{\Gamma}(\Pi)|\hspace{3mm}\alpha_{\Gamma}(\Pi)\le \epsilon\},
\end{align}
here the minimum takes over all the POVMs $\{\Pi,I-\Pi\}.$ \\
\indent Then we generalize the  hypothesis testing on quantum combs,
\begin{align}
\textit{Null hypothesis  } H_0: \hspace{3mm}C_0,\nonumber\\
\textit{Alternative hypothesis   } H_1: \hspace{3mm}C_1.
\end{align}
We define a quantity on quantum combs similar to $\beta_{\epsilon}(\rho||\sigma),$
\begin{align}
\beta_{\epsilon}(C_0||C_1)=&\min\{\beta_{\Gamma}(\Pi)|\alpha_{\Gamma}(\Pi)\le \epsilon\},\nonumber\\
\beta_{\Gamma}(\Pi)=&\{\tr[\Pi\Gamma^{1/2}C_1\Gamma^{1/2}]|\Gamma\in \textit{DualComb}\}\nonumber\\
\alpha_{\Gamma}(\Pi)=&\{\tr[(I-\Pi)\Gamma^{1/2}C_0\Gamma^{1/2}]|\Gamma\in \textit{DualComb}\}.
\end{align}
here we denote that $\{\Pi,I-\Pi\}$ is a POVM and 
\begin{align}\label{duc}
DualComb=\{\Gamma=&I_{A^{out}_N}\otimes \Gamma^{(N)},\nonumber\\
\tr_{A_n^{in}}[\Gamma^{(n)}]=&I_{A_{n-1}^{out}}\otimes\Gamma^{(n-1)},\hspace{3mm}n=2,3,\cdots,N.\nonumber\\
\tr_{A_1^{in}}[\Gamma^{(1)}]=&1.\}
\end{align}
\section{main results}
\indent In this section, we first present a definition similar to the smooth max-relative entropy $\tilde{D}_{\max}^{\epsilon}$ of two combs, and we present bounds of the quantum hypothesis testing in terms of quantum combs that are on the max-relative entropy of two combs. Then we present a bound of the regularized $\tilde{D}_{\max}^{\epsilon}$ of two combs with respect to the relative entropy of two combs. At last, we present the quantum asymptotic equipartition property in terms of the maximum score for quantum operators.\\
\indent When $C_0$ and $C_1$ are two quantum combs, then we present a quantity similar to the smooth max-relative entropy of $C_0$ with respect to $C_1$
	\begin{align}
\tilde{D}_{\max}^{\epsilon}(C_0||C_1)=&\max_{\Gamma}\min_{C\in \tilde{S^{^{\Gamma}}_{\epsilon}}}D_{\max}(C||\Gamma^{1/2}C_1\Gamma^{1/2})\nonumber\\
\tilde{S^{^{\Gamma}}_{\epsilon}}=\{C|\frac{1}{2}||&\Gamma^{1/2}C_0\Gamma^{1/2}-C||_1\le \epsilon,tr C=1\}\hspace{3mm}\Gamma  \in \textit{DualComb}
\end{align} 
Next we present a relation between the smooth max-relative entropy and the definiton we defined above of $C_0$ with respect to $C_1$.
\begin{theorem}
	Assume $C_0$ and $C_1$ are two quantum combs on $\otimes_{j=1}^n[\mathcal{H}_j^{out}\otimes\mathcal{H}_j^{in}],$ then we have 
	\begin{align}
	\tilde{D}^{\epsilon\prod_{j}\dim d_{j}^{out}}(C_0||C_1)\le D^{\epsilon}(C_0||C_1).
	\end{align}
\end{theorem}
\begin{proof}
	By the result in \cite{chiribella2016optimal}, we have when $E_0$ and $E_1$ are two quantum combs,
	\begin{align}
	D_{\max}(E_0||E_1)=\max_{\Gamma}D_{\max}(\Gamma^{1/2}E_0\Gamma^{1/2}||\Gamma^{1/2}E_1\Gamma^{1/2}),
	\end{align}
	here $\Gamma\in DualComb$.\par
	Assume $C$ is the optimal in terms of the  $D^{\epsilon}(C_0||C_1)$, as $C$ is a quantum comb, then $\tr\Gamma^{1/2}C\Gamma^{1/2}=1,$ and 
	\begin{align}
	&\tr|\Gamma^{1/2}C_0\Gamma^{1/2}-\Gamma^{1/2}C\Gamma^{1/2}|\nonumber\\
\le &\tr [\Gamma^{1/2}(C_0-C)_{+}\Gamma^{1/2}]\nonumber\\
\le &\epsilon\tr\Gamma=\epsilon\prod_{j}\dim d_{j}^{out}
	\end{align}
	here the second inequality is due to $\frac{1}{2}||C-C_0||\le\epsilon$. Then we have $\Gamma^{1/2}C\Gamma^{1/2}$ is in $\tilde{S}_{\prod_j \dim d_j^{out}\epsilon}^{\Gamma},$ at last, by the definition of $\tilde{D}^{\epsilon}$ and $ D^{\epsilon}(C_0||C_1)$, we finish the proof.
\end{proof}
\begin{theorem}
	Assume $C_0$ and $C_1$ are two quantum combs. Let $\lambda>0,$ $	  \triangle_{\Gamma}(\lambda)=[\Gamma^{1/2}(C_0-\lambda C_1)\Gamma^{1/2}]_{+},$ here $\Gamma\in$ $DualComb,$ then we have 
	\begin{align}
	\tilde{D}_{\max}^{g(\lambda)}(C_0||C_1)\le \max_{\Gamma}\log\frac{\lambda}{\sqrt{1-g_{\Gamma}^2(\lambda)}}, \label{qcdmaxht}
	\end{align}
	here we assume $g_{\Gamma}(\lambda)=\sqrt{\tr[\triangle_{\Gamma}(\lambda)(2-tr\triangle_{\Gamma}(\lambda))]}$ and the maximum takes over all the $\Gamma\in {DualComb}$.
\end{theorem}
\begin{proof}
		  As $C_x$, $x=0,1$ is a quantum comb, then $C_x$ is a semidefinite positive operator, that is, $\Gamma^{1/2}C_x\Gamma^{1/2}$ is also semidefinite positive,
	\begin{align}
	&tr \Gamma^{1/2} C_x\Gamma^{1/2}\nonumber\\
	=&tr\Gamma C_x\nonumber\\
	=&1\hspace{5mm} x=0,1,
	\end{align}
	the second equality is due to the Lemma $\ref{qcb}$ and $\ref{qut}$, that is, $\Gamma^{1/2}C_x\Gamma^{1/2}$ is a state.  
\par	Then we could prove the theorem by a similar method in \cite{datta2013smooth}, next we present the proof for the integrity.\par
	  Here we denote
	  \begin{align}
	  \triangle_{\Gamma}(\lambda)=[\Gamma^{1/2}(C_0-\lambda C_1)\Gamma^{1/2}]_{+},
	  \end{align}
	   then we have 
	   \begin{align}
	   \Gamma^{1/2}C_0\Gamma^{1/2}\le \Gamma^{1/2}\lambda C_1\Gamma^{1/2}+\triangle_{\Gamma}(\lambda),
	   \end{align}
	   by the Lemma C. 5 in \cite{brandao2010generalization}, there exists a state $\rho$  such that 
	   \begin{align}
	   \rho\le (1-tr\triangle_{\Gamma}(\lambda))^{-1}&\Gamma^{1/2}\lambda C_1\Gamma^{1/2},\nonumber\\
	   D_{\max}(\rho||\Gamma^{1/2}C_1\Gamma^{1/2})&\le \log[\lambda(1-tr\triangle_{\Gamma}(\lambda))^{-1}],\nonumber\\
\frac{1}{2}||\rho-\Gamma^{1/2}C_0\Gamma^{1/2}||&\le g_{\Gamma}(\lambda),
	   \end{align}
when $\triangle_{\Gamma}$ gets the maximum with the variable $\Gamma$, the function $\log[\lambda(1-tr\triangle_{\Gamma}(\lambda))^{-1}]$ gets the maximum.
then we finish the proof.
\end{proof}
\begin{remark}\label{trianglelambda}
By the same method in \cite{datta2013smooth}, we have when $\supp C_0\subset\supp C_1,$ $tr\triangle_{\Gamma}(\lambda)$ is strictly decreasing and continuous on $[0,2^{D_{\max}(\Gamma^{1/2}C_0\Gamma^{1/2}||\Gamma^{1/2}C_1\Gamma^{1/2})}]$ with range $[0,1],$ and $g_{\Gamma}(\lambda)$ is strictly decreasing and continuous on $[0,2^{D_{\max}(\Gamma^{1/2}C_0\Gamma^{1/2}||\Gamma^{1/2}C_1\Gamma^{1/2})}]$ with range $[0,1]$. 
\end{remark}
\par\indent Next we present a relationship between the max-relative entropy and the hypothesis testing in terms of quantum combs.
\begin{theorem}
	Assume $C_0$ and $C_1$ are two quantum combs, then we have 
	\begin{align}
\tilde{D}_{\max}^{f(\epsilon)}	\le-\log\beta_{1-\epsilon}(C_0||C_1)\le \tilde{D}_{\max}^{\epsilon^{'}}(C_0||C_1)+\log\frac{1}{\epsilon-\epsilon^{'}},
	\end{align}
	here $f(\epsilon)=\sqrt{1-(1-\epsilon)^2}.$
\end{theorem}
\begin{proof}
\indent First we prove the upper bound. If we could prove $\forall\Pi, 0\le \Pi\le I$ such that
\begin{align}
\log\beta_{\Gamma}(\Pi)<-\tilde{D}_{\max}^{\epsilon^{'}}(C_0||C_1)-\log\frac{1}{\epsilon-\epsilon^{'}},\label{assume}
\end{align}
then 
\begin{align}
\alpha_{\Gamma}(\Pi)>1-\epsilon.
\end{align}
 By the definition of $\tilde{D}_{\max}^{\epsilon}(C_0||C_1),$ we have there exists a state $C$ and a tester $\Gamma$ such that 
 \begin{align}
 C\le 2^{\tilde{D}_{\max}^{\epsilon^{'}}(C_0||C_1)}\Gamma^{1/2}C_1\Gamma^{1/2},\label{dmax'}
 \end{align}
 then we have
  \begin{align}
 tr\Pi  C\le& 2^{\tilde{D}_{\max}^{\epsilon^{'}}(C_0||C_1)}tr \Pi \Gamma^{1/2}C_1\Gamma^{1/2}\nonumber\\
 =&2^{\tilde{D}_{\max}^{\epsilon^{'}}(C_0||C_1)}\beta_{\Gamma}(\Pi)\nonumber\\
 <&2^{\tilde{D}_{\max}^{\epsilon^{'}}(C_0||C_1)}2^{-[{\tilde{D}_{\max}^{\epsilon^{'}}(C_0||C_1)}+\log(\epsilon-\epsilon^{'})]}\nonumber\\
 =&\epsilon-\epsilon^{'},
 \end{align}the first inequality is due to $(\ref{dmax'})$
 the second inequality is due to $(\ref{assume}).$
 Then we have 
 \begin{align}
 1-\alpha_{\Gamma}(\Pi)=&tr(\Pi \Gamma^{1/2}C_0\Gamma^{1/2})\nonumber\\
 =&tr(\Pi C)+\tr[\Pi (\Gamma^{1/2}C_0\Gamma^{1/2}-C)]\nonumber\\
 \le &\epsilon-\epsilon^{'}+|| \Gamma^{1/2}C_0\Gamma^{1/2}-C||/2\nonumber\\
 \le &\epsilon-\epsilon^{'}+\epsilon^{'}\le \epsilon,
 \end{align}
 \indent Next we show the other hand. Here we assume $\Gamma$ is the optimal in terms of $\tilde{D}^{f(\epsilon)}_{max}(C_0||C_1).$ By the remark $\ref{trianglelambda},$ then we have there exists $\lambda$ such that $tr\triangle_{\Gamma}(\lambda)=\epsilon,$ and we denote $\Pi_{\Gamma}=\{\Gamma^{1/2}C_0\Gamma^{1/2}\ge\Gamma^{1/2}\lambda C_1\Gamma^{1/2}\}$, then
 \begin{align}
 tr\Pi_{\Gamma} \Gamma^{1/2}C_0\Gamma^{1/2}\ge tr\Pi_{\Gamma}\Gamma^{1/2}(C_0-\lambda C_1)\Gamma^{1/2}=\epsilon,
 \end{align}
 that is, $\alpha_{\Gamma}(\Pi)\le 1-\epsilon,$ hence we have 
 \begin{align}
 -\log\beta_{1-\epsilon}(C_0||C_1)\ge -\log\beta_{\Gamma}(\Pi),
 \end{align}
 as 
 \begin{align}
  \epsilon=&tr\triangle(\lambda)\nonumber\\=&tr\Pi_{\Gamma}\Gamma^{1/2}(C_0-\lambda C_1)\Gamma^{1/2}\nonumber\\\le& 1-\lambda tr\Pi_{\Gamma} \Gamma^{1/2}C_1\Gamma^{1/2},
 \end{align}
 that is, $\beta_{\Gamma}(\Pi)\le \frac{1-\epsilon}{\lambda},$
then we have there exists a $\Gamma$ such that
\begin{align}
&-\log\beta_{\Gamma}(\Pi)\nonumber\\
\ge& \log\lambda-\log(1-\epsilon)\nonumber\\
\ge& \tilde{D}_{\max}^{f(\epsilon)}(C_0||C_1)+\log (\sqrt{1-g_{\Gamma}(\lambda)})-\log(1-\epsilon)\nonumber\\=&D_{\max}^{f(\epsilon)}(C_0||C_1),
\end{align}
 the second inequality is due to the continuity of $g_{\Gamma}(\lambda)$ and the range of $g_{\Gamma}(\lambda)$ is $[0,1]$ which is presented in the Remark  \ref{trianglelambda}.
\end{proof}
\par Then we present a result on the max-relative entropy of quantum combs by the generalized quantum stein's lemma \cite{brandao2010generalization}. 
\begin{theorem}
	Assume $C_0$ is a quantum comb, $C_1$ is a quantum comb with full rank, then we have 
	\begin{align}
	&\lim_{\epsilon\rightarrow 0}\lim_{k\rightarrow\infty}\frac{1}{k}\tilde{D}^{\epsilon}_{\max}(C^{\otimes k}_0||C^{\otimes k}_1)\nonumber\\
	\ge&\max_{\Gamma\in DualComb} D(\Gamma^{1/2}C_0\Gamma^{1/2}||\Gamma^{1/2} C_1\Gamma^{1/2}), \label{gsl}
	\end{align}
\end{theorem}
\begin{proof}
	Here we assume 
	\begin{align}
		M_k^{\Gamma}=\{(\Gamma^{1/2})^{\otimes k}C_1^{\otimes k}(\Gamma^{1/2})^{\otimes k}|\Gamma\in DualComb.\},
	\end{align}
\par
	Next we prove $M_k^{\Gamma}$ satisfies the properties 1-5 in \cite{brandao2010generalization}.\\
	(i) the set $\{\Gamma^{(1)}|\tr_{A^{in}_1}[\Gamma^{(1)}]=1\}$ is closed, then as the preimage of a closed set is closed, then we have the set is closed, then the set $M_1^{\Gamma}$ is closed. And by the definition of $Comb^{\otimes k}$, we have the set $M_k^{\Gamma}$ is closed.\\
	(ii)  the set $M_k^{\Gamma}$ is a set owning an operator with full rank.\\
	(iii) Assume $D_{k+1}$ is an element in the set $M_{k+1}^{\Gamma},$as $$\tr_{l}(\Gamma^{1/2})^{\otimes k+1}C_1^{\otimes k+1}(\Gamma^{1/2})^{\otimes k+1}=(\Gamma^{1/2})^{\otimes k}C_1^{\otimes k}(\Gamma^{1/2})^{\otimes k},$$
 then we will show that when $l$ is arbitrary, then $\tr_{l}D_{k+1}\in M_k $.\\
 (iv)  as \begin{align}
 (\Gamma^{1/2})^{\otimes s}C_1^{\otimes s}(\Gamma^{1/2})^{\otimes s} \in M_s^{\Gamma},\nonumber\\ (\Gamma^{1/2})^{\otimes t}C_1^{\otimes t}(\Gamma^{1/2})^{\otimes t } \in M_t^{\Gamma},
 \end{align}
 then from the definition of $M_k^{\Gamma},$ then we have 
 \begin{align}
  (\Gamma^{1/2})^{\otimes s}C_1^{\otimes s}(\Gamma^{1/2})^{\otimes s}\otimes(\Gamma^{1/2})^{\otimes t}C_1^{\otimes t}(\Gamma^{1/2})^{\otimes t }\in M^{\Gamma}_{s+t}.
 \end{align}
 (v) Assume $(\Gamma^{1/2})^{\otimes k}C_1^{\otimes k}(\Gamma^{1/2})^{\otimes k}\in M_k^{\Gamma},$ then $P_{\pi}[(\Gamma^{1/2})^{\otimes k}C_1^{\otimes k}(\Gamma^{1/2})^{\otimes k}]=(\Gamma^{1/2})^{\otimes k}C_1^{\otimes k}(\Gamma^{1/2})^{\otimes k},$ $\pi$ is a permutation.\\
 Next as
 	\begin{align}
\lim_{\epsilon\rightarrow 0}\lim_{k\rightarrow\infty}\frac{1}{k}\tilde{D}^{\epsilon}_{\max}(C^{\otimes k}_0||C^{\otimes k}_1)\ge	&\lim_{\epsilon\rightarrow 0}\lim_{k\rightarrow\infty}\max_{\Gamma\in DualComb}\min_{C\in \tilde{S^{{\Gamma}^{\otimes k}}_{\epsilon}}}\frac{1}{k} D_{\max}(C||(\Gamma^{1/2})^{\otimes k}(C_1)^{\otimes k}(\Gamma^{1/2})^{\otimes k})\nonumber\\
 	\ge& \lim_{\epsilon\rightarrow 0}\lim_{k\rightarrow\infty}\min_{C\in \tilde{S^{{\Gamma}^{\otimes k}}_{\epsilon}}}\frac{1}{k} D_{\max}(C||(\Gamma^{1/2})^{\otimes k}C^{\otimes k}_1(\Gamma^{1/2})^{\otimes k}),
 	\end{align}
then 
	\begin{align}
&\lim_{\epsilon\rightarrow 0}\lim_{k\rightarrow\infty}\max_{\Gamma\in DualComb}\min_{C\in \tilde{S^{{\Gamma}^{\otimes k}}_{\epsilon}}}\frac{1}{k} D_{\max}(C||(\Gamma^{1/2}C_1\Gamma^{1/2})^{\otimes k})\nonumber\\
\ge& \lim_{\epsilon\rightarrow 0}\lim_{k\rightarrow\infty}\min_{C\in \tilde{S^{{\Gamma}^{\otimes k}}_{\epsilon}}}\frac{1}{k} D_{\max}(C||(\Gamma^{1/2}C_1\Gamma^{1/2})^{\otimes k}),
\end{align}
and we have
 \begin{widetext}
 \begin{align}
 \lim_{\epsilon\rightarrow 0}\lim_{k\rightarrow\infty}\max_{\Gamma\in DualComb}\min_{C\in \tilde{S^{{\Gamma}^{\otimes k}}_{\epsilon}}}\frac{1}{k}  {D}_{\max}(C||(\Gamma^{1/2}C_1\Gamma^{1/2})^{\otimes k})\ge \max_{\Gamma\in DualComb}\lim_{\epsilon\rightarrow 0}\lim_{k\rightarrow\infty}\min_{C\in \tilde{S^{{\Gamma}^{\otimes k}}_{\epsilon}}}\frac{1}{k} {D}_{\max}(C||(\Gamma^{1/2}C_1\Gamma^{1/2})^{\otimes k}),\label{mutual}
 \end{align}
 \end{widetext}
 At last, by the Proposition $\uppercase\expandafter{\romannumeral2}.1$ in \cite{brandao2010generalization}, we have 
 \begin{align}
 &\max_{\Gamma\in DualComb}\lim_{\epsilon\rightarrow 0}\lim_{k\rightarrow\infty}\min_{C\in \tilde{S^{{\Gamma}^{\otimes k}}_{\epsilon}}}\frac{1}{k} {D}_{\max}(C||(\Gamma^{1/2}C_1\Gamma^{1/2})^{\otimes k})\nonumber\\
 =&\max_{\Gamma\in DualComb} \lim_{k\rightarrow\infty}\frac{1}{k}D((\Gamma^{1/2}C_0\Gamma^{1/2})^{\otimes k}||(\Gamma^{1/2} C_1\Gamma^{1/2})^{\otimes k})\nonumber\\
 =&\max_{\Gamma\in DualComb} D(\Gamma^{1/2}C_0\Gamma^{1/2}||\Gamma^{1/2} C_1\Gamma^{1/2}).
 \end{align}
\end{proof}
\par    \indent At last, we recall a scenario of assessing the performance of an unknown quantum network \cite{chiribella2016optimal}. There a quantum network is connected to a quantum tester, when the tester returns an outcome $x$, we assign a weight $w_x$, let $C$ be the quantum comb and $T=\{T_x|x\in Y\}$ be the quantum tester, then the average score is given by 
\begin{align}
w=&\sum_x w_x(T_x*C)\nonumber\\
=&\Omega*C,
\end{align}
here $\Omega=\sum_x w_x T_x,$ here we call $\Omega$ the performance operator, then we have the performance of a quantum network $C$ is determined by $\Omega$.  And the maximum score is 
\begin{align}
w_{\max}=&\max_{C\in \mathcal{C}}\Omega*C\nonumber\\
=&\max_{C\in\mathcal{C}}tr\Omega C^T\nonumber\\
=&\max_{C\in\mathcal{C}}tr \Omega C.
\end{align}
Here we denote the set $\mathcal{C}$ as the set of quantum combs.
The last equality is due to the fact that the set of quantum comb is closed under transpose.
 Then we present another characterization of $w_{\max}$ of a performance operator $\Omega.$
\begin{lemma}\label{max}\cite{chiribella2016optimal}
	Let $\Omega$ be an operator on $\otimes_{j=1}^N(\mathcal{H}_j^{out}\otimes\mathcal{H}_j^{in})$ and let $w_{max}$ be the maximum of $\tr\Omega C$ over all the elements in the set quantum comb $C$, then $w_{max}$ is given by 
	\begin{align}
	w_{max}=\min\{\lambda\in \mathbb{R}|\lambda\theta\ge \Omega,\theta\in DualComb\},
	\end{align}
	When $\Omega$ is positive, $w_{max} $ can be written as
	\begin{align}
	w_{max}=2^{D_{max}(\Omega||DualComb)}.\label{maxsc}
	\end{align}
\end{lemma}
\par \indent Next we present a result on the smooth asymptotic version for a performance operator, then we first present a lemma, this result will be used in Theorem $\ref{asy}$.
 \begin{lemma}\label{cdc}
	Assume $P(M_1,M_2)\le \epsilon,$ $M_1,M_2$ are positive operators with the same trace,  then $|w_{\max}(M_1)- w_{\max}(M_2)|\le d_{2n-2}d_{2n-4}\cdots d_{0}\epsilon$, here $d_{2i-2}$ is the dimension of the Hilbert space $A_i^{in}.$ 
\end{lemma}
\begin{proof}
	As $||M_1-M_2||_1\le 2\epsilon,$ and $\tr(M_1-M_2)=0,$ ($\tr(M_1-M_2)_{+}=\tr(M_1-M_2)_{-}$) then $\epsilon\triangle=(M_1-M_2)_{+}\le \epsilon I,$ here $(M_1-M_2)_{+}$ is the positive part of $M_1-M_2,$ $\triangle$ is a bona fide state.
	\par \indent 	Assume $\Theta_1$ is the optimal in terms of $(\ref{maxsc})$ for $M_1$ in terms of $w_{\max}(\cdot),$ then 
	\begin{align}
	\lambda_1\Theta_1+\epsilon I\ge &
	\lambda_1\Theta_1+(M_2-M_1)_{+}\nonumber\\
	\ge& M_1+(M_2-M_1)_{+}\ge M_2,
	\end{align}
	as  $DualComb$ is the set of operators satisfying $(\ref{qt})-(\ref{qt2}),$ and the operations in $Dualcomb$ are linear, then the set of $DualComb$ is convex,  and $I/(d_{2n-2}d_{2n-4}\cdots d_{0})\in DualComb, $ then we have $|w_{\max}(M_2)-w_{\max}(M_1)|\le d_{2n-2}d_{2n-4}\cdots d_{0}\epsilon.$ 
\end{proof}
\begin{theorem}\label{asy}
	Let $\Omega$ be a positive operator on $\otimes_{j=1}^N(\mathcal{H}_j^{out}\otimes\mathcal{H}_j^{in})$, $\lim_{\epsilon\rightarrow0}\lim_{n\rightarrow\infty}\frac{1}{n}\log w_{max}^{\epsilon}(\Omega^{\otimes n})=\log w_{max}(\Omega)$.
	Here
	\begin{align} 
	w_{max}^{\epsilon}(\Omega)=&\sup_{\mbox{\tiny$\begin{array}{c}1/2||\Omega-\Omega^{'}||_1\le \epsilon\\
	\tr\Omega=\tr\Omega^{'}\end{array}$}}w_{max}(\Omega^{'})\nonumber\\=&\sup_{\mbox{\tiny$\begin{array}{c}1/2||\Omega-\Omega^{'}||_1\le \epsilon\\
	\tr\Omega=\tr\Omega^{'}\end{array}$}}\min_{\theta\in DualComb}\min\{\lambda\in \mathbb{R}|\lambda\theta\ge \Omega^{'}\}, 
	\end{align} 
\end{theorem}
 \begin{proof}
 	Here we will show that $\lim\limits_{n\rightarrow\infty}\frac{1}{n}\log w_{max}^{\epsilon}(\Omega^{\otimes n})\ge \log w_{max}(\Omega).$\\
 	\indent From the definition and $0\le \epsilon$, we have 
 	\begin{align}
 	\log w^0_{\max}(\Omega)=&\log w_{\max}(\Omega)
 	\end{align}
 	\begin{align}\label{wmax}
 	w_{max}(\Omega)\nonumber=&\inf\{\lambda\in\mathbb{R}|\lambda\theta\ge \Omega,\theta\in DualComb\}\nonumber\\
 	=&\sup\{\tr\Omega M|M\in \mathcal{C}\},
 	\end{align}
 	Here the first equality is due to the Lemma \ref{max}, the second equality is due to the definition of $w_{max}.$\\
 	\indent	By the definition of comb, it should be closed under the operation tensor, as  if $\Theta\in Comb,$ then $\Theta^{\otimes n}$ satisfies the equality $(\ref{qc'}),$ that is, $\Theta^{\otimes n}\in Comb.$ Next we assume $M$ is the optimal in terms of the second equality $(\ref{wmax})$ of $w_{\max}(\Omega)$, then from the second equality of (\ref{wmax}), we have
 	\begin{align}\label{wmaxn}
 	w_{\max}(\Omega^{\otimes n})\ge& \tr\Omega^{\otimes n}M^{\otimes n}\nonumber\\=&(\tr\Omega M)^n=w_{\max}(\Omega)^n.  
 	\end{align}
 	The first inequality in $(\ref{wmaxn})$ holds is due to that the comb set is closed under the tensor operation. \\
 	Then we have 
 	\begin{align}
 	w_{\max}^{\epsilon}(\Omega^{\otimes n})\ge &w^0_{\max}(\Omega^{\otimes n})\ge w_{\max}(\Omega)^n,
 	\end{align}    \begin{align}
 	\lim_{n\rightarrow\infty}\frac{1}{n}\log w_{max}^{\epsilon}(\Omega^{\otimes n})\ge &\log w_{\max}(\Omega)
 	\end{align}
 	
 	Next we prove 
 	\begin{align}
 	\lim\limits_{\epsilon\rightarrow 0}\lim_{n\rightarrow\infty}\frac{1}{n}\log w_{\max}^{\epsilon}(\Omega^{\otimes n})\le \log w_{\max}(\Omega).\label{wmaxg}
 	\end{align}
 	
 	As \begin{align}\label{max1}
 	w_{\max}(\Omega)=\inf\{\lambda\in\mathbb{R}|\lambda\theta\ge\Omega,\theta\in DualComb\},
 	\end{align}
 	let $\theta$ be the optimal $w_{\max}(\Omega)$ in terms of (\ref{max1}), as $\theta\in DualComb,$ by the equailty (\ref{duc}), $\theta$ can be written as $I\otimes \Gamma,$ then 
 	\begin{align}
 	&(w_{\max}(\Omega) \theta)^{\otimes 2}-\Omega^{\otimes 2}\nonumber\\
 	=&(w_{max}I\otimes\Gamma)^{\otimes 2}-\Omega\otimes(w_{\max} I\otimes\Gamma)\nonumber\\+&\Omega\otimes(w_{\max} I\otimes\Gamma)-\Omega^{\otimes2}\nonumber\\
 	=&(w_{\max}I\otimes\Gamma-\Omega)\otimes(w_{\max}I\otimes\Gamma)\nonumber\\+&\Omega\otimes(w_{\max}I\otimes\Gamma-\Omega)\ge0,
 	\end{align}
 	similarly, we have $(w_{\max}(\Omega) \theta)^{\otimes n}-\Omega^{\otimes n},\forall n,$ that is, $w_{\max}(\Omega^{\otimes n})\le w_{\max}^n(\Omega) ,$ by the Lemma \ref{cdc}, we have 
 	\begin{align}
 	&\frac{1}{n}\log w_{\max}^{\epsilon}(\Omega^{\otimes n})\nonumber\\
 	\le &\frac{1}{n}\log(d_{2n-2}d_{2n-4}\cdots d_{0}\epsilon+w_{\max}(\Omega^{\otimes n})\nonumber\\
 	\le &\frac{1}{n}\log(d_{2n-2}d_{2n-4}\cdots d_{0}\epsilon+w_{\max}^n(\Omega)),
 	\end{align}
 	then we have
 	\begin{align}
 	\lim\limits_{\epsilon\rightarrow 0}\lim_{n\rightarrow\infty}\frac{1}{n}\log w_{\max}^{\epsilon}(\Omega^{\otimes n})\le \log w_{\max}(\Omega).
 	\end{align}
 \end{proof}	 
\section{Conclusion}
\indent  In this article, we have presented a quantity with respect to the smooth max-entropy relative for two quantum combs, we have also presented a relationship between the max-relative entropy and the type \uppercase\expandafter{\romannumeral2} error of quantum hypothesis testing in terms of quantum networks we presented, next we have presented a lower bound of the regularized smooth max-entropy relative to two quantum combs by the result in \cite{brandao2010generalization}. At last, we have shown a smooth asymptotic version of the performance of a quantum positive operator is still the performance of the operator.

\section{Acknowledgements}
The authors thank Mengyao Hu, she gave us help on plotting the type of the network.  Authors were supported by the NNSF of China (Grant No. 11871089), and the Fundamental Research Funds for the Central Universities (Grant Nos. KG12080401 and ZG216S1902).
\bibliographystyle{IEEEtran}
\bibliography{reff}
\end{document}